\documentclass[letterpaper, 10 pt, conference]{ieeeconf}  

\IEEEoverridecommandlockouts                              
\usepackage{graphicx} 
\usepackage{amssymb,amsfonts}
\usepackage{amsmath}
\usepackage{algorithmic}
\usepackage{algorithm}
\usepackage{textcomp}
\usepackage{xcolor}
\usepackage{multirow}
\usepackage[T1]{fontenc}
\usepackage[utf8]{inputenc}
\usepackage{textcomp}
\usepackage{booktabs}
\usepackage{multirow}
\usepackage{balance}

\newtheorem{corollary}{Corollary}[section]
\newtheorem{lemma}{Lemma}[section]
\newtheorem{proposition}{Proposition}[section]
\newtheorem{remark}{Remark}

\newtheorem{assumption}{Assumption}

\DeclareMathOperator{\col}{col}
\DeclareMathOperator{\blkdiag}{blkdiag}

\title{\LARGE \bf
Delay-Robust Primal-Dual Dynamics for Distributed Optimization
}

\author{G\"{o}k\c{c}en Devlet \c{S}en, Juan E. Machado, G\"ulay \"Oke G\"unel and Johannes Schiffer
\thanks{This research is supported by the German Federal Government, the Federal Ministry of Research, Technology and Space and the State of Brandenburg within the framework of the
joint project EIZ: Energy Innovation Center (project numbers 85056897 and 03SF0693A) with funds from the Structural Development Act (Strukturst\"{a}rkungsgesetz) for coal-mining regions.}
\thanks{G.D. \c{S}en, J.E. Machado and J. Schiffer are with Brandenburg University of Technology Cottbus-Senftenberg, 03046 Cottbus, Germany
        {\tt\small \{sen, machadom, schiffer\}@b-tu.de}}%
\thanks{G.D. \c{S}en and G.\"O. G\"unel are with Istanbul Technical University, 34469 Istanbul, Turkey
{\tt\small \{sengo, gulay.oke\}@itu.edu.tr}}%
\thanks{J. Schiffer is with Fraunhofer IEG, Fraunhofer Research Institution for Energy Infrastructures and Geotechnologies IEG, 03046 Cottbus, Germany {\tt\small johannes.schiffer@ieg.fraunhofer.de}}%
}

\begin{document}

\maketitle
\thispagestyle{empty}
\pagestyle{empty}


\begin{abstract}
Continuous-time primal-dual gradient dynamics (PDGD) is an ubiquitous approach for dynamically solving constrained distributed optimization problems. Yet, the distributed nature of the dynamics makes it prone to communication uncertainties, especially time delays. To mitigate this effect, we propose a delay-robust continuous-time PDGD. The dynamics is obtained by augmenting the standard PDGD with an auxiliary state coupled through a gain matrix, while preserving the optimal solution. Then, we present sufficient tuning conditions for this gain matrix in the form of linear matrix inequalities, which ensure uniform asymptotic stability in the presence of bounded, time-varying delays. The criterion is derived via the Lyapunov-Krasovskii method. A numerical example illustrates the improved delay robustness of our approach compared to the standard PDGD under large, time-varying delays.
\end{abstract}

\section{Introduction}
Distributed optimization algorithms enable multiple agents to collaboratively solve a global optimization problem using local information and communication with neighboring agents~\cite{ nedic2018distributed,yang2019survey}. 
Each agent holds part of the objective or the constraints, performs local computations, and exchanges limited information with its neighbors. 
The goal is to obtain a globally optimal solution without relying on a centralized coordinator with full problem knowledge. This framework is particularly relevant for large scale and geographically distributed systems, where centralized architectures face limitations in scalability, communication, privacy, and robustness~\cite{ nedic2018distributed,yang2019survey}. Distributed approaches are strongly motivated in applications such as economic dispatch~\cite{liu2019distributed,schiffer2017robustness},  optimal power flow~\cite{dall2013distributed}, resource allocation~\cite{xu2017distributed}.

Distributed optimization algorithms rely on information exchange among agents to coordinate updates and satisfy global constraints. In practice, however, communication imperfections, such as delays, are unavoidable due to the physical separation of agents and network limitations. If such delays are not explicitly considered in the algorithm design, they can degrade performance and lead to undesirable behavior such as oscillations or even instability~\cite{yang2016distributed,wang2018distributed}.

Several studies have analyzed the effect of communication delays on the stability of distributed optimization algorithms. In \cite{yang2016distributed}, a Lyapunov-Krasovskii Functions (LKF) based stability analysis was conducted for a consensus-based algorithm with homogeneous time-varying delays. This approach was extended in \cite{wang2018distributed} to constrained distributed optimizations. In both works, delay-dependent stability conditions are expressed as linear matrix inequalities (LMIs). In \cite{wang2019distributed}, the exponential stability of a distributed resource allocation algorithm is proven under homogeneous delays, with additional restrictions imposed in the case of large delays. In \cite{li2020smooth}, a constrained distributed optimization problem is studied and the stability is established under heterogeneous communication delays using a passivity-based approach.

The above mentioned works analyze stability of distributed optimization algorithms under communication delays for fixed parameters, but do not provide a systematic synthesis that guarantees stability for given delay bounds.
In this paper, we propose an augmented primal-dual gradient dynamics (PDGD) for constrained distributed optimization subject to heterogeneous, time-varying communication delays. PDGD, also known as saddle-point dynamics, is a gradient-based approach for solving convex optimization problems. Early forms of PDGD were introduced in \cite{kose1956solutions} and \cite{arrow1958studies}. Since then, its convergence properties have been extensively investigated. Asymptotic stability results are established in \cite{feijer2010stability,cherukuri2016asymptotic}, while exponential convergence is shown in \cite{qu2018exponential,guo2022exponential}.
However, in the presence of communication delays, the standard PDGD may exhibit oscillations or even instability. 
To address this issue, we augment the PDGD with an auxiliary state coupled through a tunable gain matrix that provides an additional degree of freedom to systematically shape the closed-loop dynamics and improve delay robustness.

Our main contributions are three-fold: (i) a dynamic augmentation of the PDGD with a gain matrix as an additional design parameter that enables a systematic synthesis while preserving the optimal solution;
(ii) a delay-dependent, LMI-based tuning condition that guarantees uniform asymptotic stability for given heterogeneous time-varying delays and accommodates three types of convex local cost functions; 
(iii) a numerical comparison with the standard PDGD. 

\textbf{Notation.}
We define the set $\mathbb{R}_{\geq 0}=\{x\in\mathbb{R}|x\geq0\}$.
For a matrix $A \in \mathbb{R}^{n\times n}$, the notation $A>0$ denotes that $A$ is positive definite. 
The elements below the diagonal of a symmetric matrix are denoted by $\ast$.
The matrices $I_n$ and $0_{n\times m}$ denote the $n\times n$ identity matrix and the $n\times m$ zero matrix, respectively. The indices are omitted when the dimensions are clear from the context. For a finite set $\mathcal N$, $|\mathcal N|$ denotes its cardinality. For a set of positive natural numbers $\mathcal N$, $x:=\col(x_i)_{i \in \mathcal{N}}$ denotes the stacked vector with the entries $x_i\in\mathbb R^{n_i}$. Likewise, $A=\blkdiag(A_i)_{i \in \mathcal{N}}$ denotes a block-diagonal matrix with diagonal blocks $A_i$. For block matrices, $A=[A_{ij}]$ denotes the matrix composed of the block entries $A_{ij}$.   
For a time-delayed signal $x(t-\tau(t))$ with delay $\tau(t)\geq0$, we use the short-hand $x_\tau=x(t-\tau(t))$.

\section{Problem Statement}
Consider a multi-agent system with $N \geq 2$ agents whose objective is to solve the following optimization problem
\begin{subequations}\label{eq:opt_prob}
\begin{align}
    \min_{x\in \mathbb{R}^{n}} f(x)&=\sum_{i=1}^N f_i(x_i), \\
    \text{s.t. }& Ax=b,  \label{eq:coupling_const}
\end{align} 
\end{subequations}
where $x=\col(x_i)_{i \in \mathcal{N}}\in\mathbb{R}^{n}$ is the decision variable of~\eqref{eq:opt_prob} with $x_i \in \mathbb{R}^{n_i}$, $\sum_{i=1}^{N} n_i =n$ and $\mathcal{N}=\{1,\ldots, N\}$. The function $f_i: \mathbb{R}^n\rightarrow \mathbb{R}$ is the local cost of agent $i$, the vector $b\in \mathbb{R}^m$ and the matrix $A\in\mathbb{R}^{m\times n}$ with $m<n$ incorporate the coupling constraints among the agents. Problem \eqref{eq:opt_prob} appears in several applications, including resource allocation, economic dispatch, and consensus-type coordination problems~\cite{nedic2018distributed,yang2019survey}. We make the following standard assumptions on~\eqref{eq:opt_prob}~\cite{qu2018exponential, wang2011control}.
\begin{assumption}\label{eq:cost_convexity}
    The functions $f_i$ in \eqref{eq:opt_prob}, $\forall \, i \in \mathcal{N}$, are twice continuously differentiable, $\mu_i$-strongly convex (with~$\mu_i>0$), and $\ell_i$-smooth, 
    i.e., they satisfy
    \begin{equation*}
        \mu_i \left\|x_i - y_i \right\|^2 \leq \langle \nabla f_i(x_i) -\nabla f_i(y_i), x_i - y_i \rangle \leq \ell_i\left\|x_i - y_i \right\|^2
    \end{equation*}
    for all $x_i, y_i \in \mathbb{R}^{n_i}$. 
\end{assumption}
\begin{assumption}\label{eq:rank_cond}
    The matrix $A$ in \eqref{eq:opt_prob} has full row rank. 
\end{assumption}
The Lagrangian, associated with~\eqref{eq:opt_prob}, is given as follows
\begin{equation*}
    L(x,\lambda) = f(x) + \lambda^\top (Ax-b),
\end{equation*}
where $\lambda \in \mathbb{R}^m$ is the vector of Lagrange multipliers. The corresponding Karush-Kuhn-Tucker (KKT) conditions are given by 
\begin{subequations}\label{eq:KKT_cond}
\begin{align}
    \frac{\partial L({x}^\ast,{\lambda}^\ast)}{\partial x} &= \nabla f({x}^\ast) + A^\top {\lambda}^\ast = 0, \\
     \frac{\partial L({x}^\ast,{\lambda}^\ast)}{\partial \lambda} &= A {x}^\ast - b = 0,
\end{align}
\end{subequations}
and provide necessary and sufficient conditions for optimality~\cite{boyd2004convex}. Assumption~\ref{eq:cost_convexity}-\ref{eq:rank_cond} ensures that the saddle point $({x}^\ast, {\lambda}^\ast)$ exists and that is the unique primal-dual optimizer~\cite{wang2011control, qu2018exponential}. To calculate $({x}^\ast, {\lambda}^\ast)$, one could perform the PDGD~\cite{wang2011control}. 

In a distributed setting, the PDGD is usually implemented such that each agent $i$ updates its local primal variable and 
its local dual variable through communication with neighboring agents, e.g., via a consensus equality constraint~\cite{chang2014distributed, jakovetic2020primal}. 
In this paper, we consider a slightly more general type of constraint by assuming that $A = [A_{ij}]$ in \eqref{eq:opt_prob} is a block matrix composed of $M\geq 1$ block rows and $N\geq2$ block columns with $M \leq N$ and $b=\col(b_1, \ldots, b_M)$ to be a vector partitioned into $M$ blocks, where $b_i \in \mathbb{R}^{m_i}$ with $\sum_{i=1}^M m_i = m$.  
With the $i$-th block row in~\eqref{eq:coupling_const} we associate a dual variable $\lambda_i \in \mathbb{R}^{m_i}$ and denote the index set of agents that maintain dual variables by $\mathcal{N}_\lambda = \{1,\ldots, M\}$.
The considered distributed PDGD
is given by 
\begin{subequations}\label{eq:PD_distributed}
\begin{align}
\dot{\hat{x}}_i &= -\nabla f_i(\hat{x}_i)
- \sum_{j=1}^M
A_{ji}^\top\,\hat{\lambda}_j, \quad \forall \, i \in \mathcal{N}, \\
\dot{\hat{\lambda}}_i &= \sum_{j =1}^N
A_{ij}\,\hat{x}_j
- b_i,
\quad \forall \, i \in \mathcal{N}_\lambda, 
\end{align}
\end{subequations}
where $\hat{x}_i \in \mathbb{R}^{n_i}$ and $\hat{\lambda}_i \in \mathbb{R}^{m_i}$ denote the primal and dual states of the agent~$i$.

Note that the case $M < N$ arises naturally in several applications, including resource allocation~\cite{xu2017distributed}, economic dispatch~\cite{liu2019distributed} and consensus problems with nonredundant constraints \cite[Chapter 9]{FB-LNS}.  
Under Assumptions~\ref{eq:cost_convexity} and~\ref{eq:rank_cond}, for any initial value of $\hat{x}_i$, $\forall \, i\in\mathcal{N}$ and $\hat{\lambda}_i$, $\forall \, i\in\mathcal{N}_\lambda$, the solutions of the PDGD~\eqref{eq:PD_distributed} converge to the unique saddle point $(x^\ast, \lambda^\ast)$ with $x^\ast=\col(x_i^\ast)_{i\in\mathcal{N}}$ and $\lambda^\ast=\col(\lambda_i^\ast)_{i\in\mathcal{N}_\lambda}$~\cite{wang2011control}.

Yet, in any practical setting the information exchange between agents executing the distributed PDGD~\eqref{eq:PD_distributed} will be affected by time delays~\cite{yang2016distributed,li2020smooth}.
When communication delays are present,~\eqref{eq:PD_distributed} becomes
\begin{subequations}\label{eq:PD_delayed}
\begin{align}
\dot{\hat{x}}_i &= -\nabla f_i(\hat{x}_i) 
- A_{ii}^\top \hat{\lambda}_i
- \sum_{\substack{j=1, j \neq i}}^M
A_{ij}^\top \,\hat{\lambda}_{j,\tau_{ij}}, \quad \forall \, i \in \mathcal{N}, \\
\dot{\hat{\lambda}}_i &= A_{ii}\hat{x}_i 
+ \sum_{\substack{j =1,  j \neq i}}^N
A_{ij}\,\hat{x}_{j,\tau_{ij}}
- b_i, 
\quad \forall \, i \in \mathcal{N}_\lambda, 
\end{align}
\end{subequations}
where $\tau_{ij}:\mathbb{R}_{\geq0}\rightarrow[0,h_{ij}]$ denotes the time-varying comunication delay between agents $j$ and $i$ with $h_{ij} > 0$ and we recall the short-hand $x_{\tau_{ij}}=x(t-\tau_{ij}(t))$.
    For slowly varying delays, $\left| \dot{\tau}_{ij} \right| \leq d_{ij}$ with $d_{ij} \in [0, 1)$, while for fast-varying delays no restrictions are imposed on $\dot{\tau}_{ij}$.
    
Communication delays may cause oscillations or instability in the standard PDGD~\eqref{eq:PD_delayed}~\cite{yang2016distributed, wang2019distributed}, motivating the need for improved delay-robust algorithms. We address the problem of enhancing the delay robustness of~\eqref{eq:PD_delayed} by proposing an augmented PDGD in the next section.

\section{An Augmented Primal-Dual Gradient Dynamics for Enhanced Delay-Robustness}
\subsection{Augmented PDGD}
We augment the standard PDGD~\eqref{eq:PD_delayed} as follows
\begin{subequations}\label{eq:proposed_dyn}
\begin{align}
    \dot{\hat{x}}_i = &-\nabla f_i(\hat{x}_i) - A_{ii}^\top\hat{\lambda}_i - \sum_{\substack{j=1, j \neq i}}^M A_{ji}^\top \hat{\lambda}_{j,\tau_{ij}} \nonumber\\
    &+ \kappa_{11,i}(\hat{x}_i-u_{1,i}) + \kappa_{12,i}(\hat{\lambda}_i-u_{2,i}), \ \forall \, i \in \mathcal{N}, \\
     \dot{\hat{\lambda}}_i =& A_{ii}\hat{x}_i + \sum_{\substack{j=1, j \neq i}}^N A_{ij} \hat{x}_{j,\tau_{ij}} - b_i  \nonumber \\
     &+ \kappa_{21,i}(\hat{x}_i-u_{1,i}) + \kappa_{22,i}(\hat{\lambda}_i-u_{2,i}), \ \forall \,i \in \mathcal{N}_\lambda,  \\
     \dot{u}_{1,i} &= \hat{x}_{i} - u_{1,i}, \quad \forall \, i \in \mathcal{N}, \\
     \dot{u}_{2,i} &= \hat{\lambda}_{i} - u_{2,i}, \quad \forall \, i \in \mathcal{N}_\lambda, 
\end{align}
\end{subequations}
where $u_{1,i} \in \mathbb{R}^{n_i}$ and $u_{2,i} \in \mathbb{R}^{m_i}$ denote additional controller states. 
The matrices 
$\kappa_{11,i} \in \mathbb{R}^{n_i \times n_i}$, 
$\kappa_{12,i} \in \mathbb{R}^{n_i \times m_i}$, 
$\kappa_{21,i} \in \mathbb{R}^{m_i \times n_i}$, and 
$\kappa_{22,i} \in \mathbb{R}^{m_i \times m_i}$ 
are the block components of the gain matrix $K_i=[\kappa_{\ell j,i}]$, $j=1,2,$ $\ell=1,2,$ which in the sequel is determined by solving an optimization problem.

\subsection{Gain Synthesis with Guaranteed Delay-Robust Stability} \label{sec:Convergence_analysis}
In this section, we provide delay-dependent conditions that guarantee uniform asymptotic stability of the equilibrium of~\eqref{eq:proposed_dyn}. 
Let $(\bar{x}, \bar{\lambda}, \bar{u}^1, \bar{u}^2)$ with $\bar{x}=\col(\bar{x}_i)_{i \in \mathcal{N}}$, $\bar{\lambda}=\col(\bar{\lambda}_i)_{i \in \mathcal{N}_\lambda}$, $\bar{u}^1=\col(\bar{u}_{1,i})_{i \in \mathcal{N}}$ and $\bar{u}^2=\col(\bar{u}_{2,i})_{i \in \mathcal{N}_\lambda}$ be the equilibrium of~\eqref{eq:proposed_dyn}. We propose the following lemma.
\begin{lemma}
    The point $(\bar{x}, \bar{\lambda}, \bar{u}^1, \bar{u}^2)$ 
    is an equilibrium of~\eqref{eq:proposed_dyn} if and only if the pair $(\bar{x}, \bar{\lambda})$ satisfies the KKT conditions~\eqref{eq:KKT_cond}.
\end{lemma}
\begin{proof}
The proof follows directly from the observation that for~\eqref{eq:proposed_dyn} it holds in equilibrium that $\dot{u}^1=0 \iff \bar{u}^1=\bar{x}$ and $\dot{u}^2=0 \iff \bar{u}^2=\bar{\lambda}$.
\end{proof}

Next, we recap the following lemma from~\cite{qu2018exponential}.
\begin{lemma}[Positive definiteness \cite{qu2018exponential}]
    Under Assumption \ref{eq:cost_convexity}, for any $\hat{x}_i \in \mathbb{R}^{n_i}$, there exists a symmetric matrix ${B}_i(\hat{x}_i)$, satisfying $\mu_i I \preceq {B}_i(\hat{x}_i) \preceq \ell_i I$, s.t. $\nabla f_i(\hat{x}_i) - \nabla f_i(\bar{x}_i) = {B}_i(\hat{x}_i)(\hat{x}_i-\bar{x}_i)$, $\forall \, i \in \mathcal{N}$.
    \label{eq:lemma1}
\end{lemma}

To guarantee stability, we introduce the error variables $\tilde{x}_i=\hat{x}_i-\bar{x}_i$, $\tilde{\lambda}_i=\hat{\lambda}_i-\bar{\lambda}_i$, $\tilde{u}_{1,i}={u}_{1,i}-\bar{u}_{1,i}$ and $\tilde{u}_{2,i}={u}_{2,i}-\bar{u}_{2,i}$
and the matrices
\[
A_i(\hat{x}_i)=
\begin{bmatrix}
- B_i(\hat{x}_i) & -A_{ii}^\top\\
A_{ii} & 0
\end{bmatrix} , \
T_{ij}=
\begin{bmatrix}
0 & -A_{ji}^\top\\
A_{ij} & 0
\end{bmatrix}.
\]
Let $\mathcal{T}_{ij}\in\mathbb{R}^{r\times r}$, $r=m+n$, denote the block matrix whose $(i,j)$-th block equals $T_{ij}$ and whose remaining blocks are zero.
To avoid double indices, introduce the index set
\[
\mathcal E = \{k=(i,j)\in\mathcal N\times\mathcal N : j\neq i,\; T_{ij}\neq0\}.
\]
Let $\rho = |\mathcal E|$
and define
\[
\mathcal T_k=\mathcal T_{i j}, \quad
\tau_k=\tau_{i j},\quad k=1,\ldots,\rho.
\]
Then, using Lemma~\ref{eq:lemma1}, the error dynamics of~\eqref{eq:proposed_dyn} can be written in a compact form as
\begin{subequations}\label{eq:compact_shifted}
\begin{align}
\dot{\tilde z} &= (\mathcal A(\hat{x})+\mathcal K)\tilde z 
+ \sum_{k=1}^{\rho} \mathcal T_k \tilde z_{\tau_k}
- \mathcal K \tilde u, \label{eq:z_dynamics} \\
\dot{\tilde u} &=\tilde z - \tilde u,
\end{align}
\end{subequations}
where $\tilde z=\col(\tilde z_i)_{i \in \mathcal{N}}$ and
$\tilde u=\col(\tilde u_i)_{i \in \mathcal{N}}$ with $\tilde{z}_i = \col(\tilde{x}_i,\tilde{\lambda}_i)$, 
$\tilde{u}_i = \col(\tilde{u}_{1,i},\tilde{u}_{2,i}) \in \mathbb{R}^{r_i\times r_i}$ with $r_i=n_i+m_i$,  and 
$\mathcal K = \blkdiag(K_i)_{i \in \mathcal{N}}$,
$\mathcal A(\hat{x}) = \blkdiag(A_i(\hat{x}_i))_{i \in \mathcal{N}}\in~\mathbb{R}^{r\times r}$.

Next, we provide delay-dependent sufficient conditions for uniform asymptotic stability of the origin of~\eqref{eq:compact_shifted}, while simultaneously enabling the synthesis of the block-diagonal gain matrix $\mathcal{K}$ (see Appendix~\ref{app:proof_1} for the proof).
\begin{proposition}\label{prop:stability_cond}
    Consider the system~\eqref{eq:compact_shifted} with Assumptions~\ref{eq:cost_convexity}-\ref{eq:rank_cond}. Let $h_{k}\geq 0$, $d_{k}\in [0,1)$, $\forall \, k \in \{1,\dots,\rho\}$ and the constant scalar $\varepsilon>0$ be given. Assume that there exist
    $r \times r$ matrices $Y_{11}$, $Y_{12}$, $Y_{22}$, $R_{k}>0$, $Q_{k}>0$, $S_{k}>0$, $S_{12,k}$, $\forall \, k \in \{1,\dots,\rho\}$,
    $P_{2} = \blkdiag(P_{2,i})_{i\in \mathcal{N}}$ with $P_{2,i}\in \mathbb{R}^{r_i\times r_i}$, $P_{2,i}+P_{2,i}^\top>0,$ and $X = \blkdiag(X_{i})_{i\in \mathcal{N}}$ with $X_{i}\in \mathbb{R}^{r_i\times r_i}$,
    such that the following matrix inequalities are feasible
    \begin{subequations}
    \begin{align}
   \hspace{-.235cm} \Phi (\hat{x}) = \begin{bmatrix}
        \phi_{11} & \phi_{12} & \phi_{13} & 0_{r\times \rho r} & 0_{r\times \rho r} \\
         \ast & \phi_{22}(\hat{x}) & \phi_{23}(\hat{x}) & \phi_{24} & \phi_{25} \\ 
         \ast & \ast & \phi_{33} & 0_{r\times \rho r} & \phi_{35} \\
         \ast & \ast & \ast & \phi_{44} & \phi_{45} \\
         \ast & \ast & \ast & \ast & \phi_{55} \\
    \end{bmatrix} <0, \label{eq:phi_matrix}
\end{align}
and
\begin{align}
    \left[\begin{array}{cc}
        R & S_{12} \\
        \ast & R
    \end{array} \right] \geq 0, \quad  \left[\begin{array}{cc}
        Y_{11} & Y_{12} \\
        \ast & Y_{22}
    \end{array} \right] >0, \label{eq:jensen_const}
\end{align}
where 
\begin{equation}\label{eq:phi_expanded}
\begin{aligned}
 R&=\blkdiag(R_{k})_{k\in\{1,\ldots,\rho\}}, \ S=\blkdiag(S_{k})_{k\in\{1,\ldots,\rho\}}, \\ S_{12}& =\blkdiag(S_{12,k})_{k\in\{1,\ldots,\rho\}}, \ {\phi}_{11} = -2 {Y}_{22}, \\ 
 {\phi}_{12} &= -{Y}_{12}^\top +{Y}_{22} -X^\top, \ {\phi}_{13} =  {Y}_{12}^\top -\varepsilon X^\top, \\
    {\phi}_{22}&(\hat{x}) = {Y}_{12} + {Y}_{12}^\top + {P}_{2}^\top\mathcal{A}(\hat{x})+ X +  \mathcal{A}(\hat{x})^\top{P}_{2} \\
    &+X^\top + \sum_{k=1}^\rho\left( {S}_{k} + {Q}_{k} - {R}_{k} \right), \\
    {\phi}_{23}&(\hat{x})= {Y}_{11} -{P}_{2}^\top + \varepsilon\mathcal{A}(\hat{x})^\top{P}_{2}+ \varepsilon X^\top,  \\
    {\phi}_{24} &= \big[\, {S}_{12,1}, \ldots, {S}_{12,\rho} \,\big],\\
    {\phi}_{25} &= \big[\,{R}_{1}-{S}_{12,1}+{P}_2^\top\mathcal{T}_{1}, \ldots, {R}_{\rho}-{S}_{12,\rho}+{P}_2^\top\mathcal{T}_{\rho}\,\big], \\
    {\phi}_{33} &= -\varepsilon{P}_{2}-\varepsilon{P}_{2}^\top +  \sum_{k=1}^\rho h_{k}^2 {R}_{k}, \\
     {\phi}_{35} &= \big[\, \varepsilon{P}_2^\top \mathcal{T}_{1}, \ldots, \varepsilon{P}_2^\top \mathcal{T}_{\rho} \,\big], \\
     {\phi}_{44} &= -S-R, \quad {\phi}_{45} = R -  S_{12}^\top, \\
     {\phi}_{55} &= -2{R} +{S}_{12} +{S}_{12}^\top -\blkdiag((1-d_{k}){Q}_{k})_{ k \in \{1,\dots,\rho\}}.
\end{aligned}
\notag 
\end{equation} 
\end{subequations}
Choose $\mathcal{K}=(P_2^\top)^{-1} X$.
Then, the origin is uniformly asymptotically stable for all delays $\tau_{k}(t)\in [0, h_{k}]$ with $\left|\dot{\tau}_{k}(t)\right|\leq d_{k}$, $\forall k \in \{1,\dots,\rho\}$. If the above conditions are feasible for $d_k=1,$ respectively $Q_{k}=0_{r\times r}$,  then the origin is uniformly asymptotically stable for all fast-varying delays $\tau_{k}(t)\in [0, h_{k}]$. 
\end{proposition}    

\subsection{State-Independent Reformulation}
Due to the presence of $\mathcal{A}(\hat{x})$ in~\eqref{eq:phi_matrix}, the matrix $\Phi(\hat{x})$ is state-dependent, which depending on the structure of $\Phi(\hat{x})$ might result in an expensive numerical implementation. Therefore, we derive next a state-independent LMI reformulation of Proposition~\ref{prop:stability_cond}.
The state dependency of $\mathcal{A}(\hat{x})$ (or the local matrices $A_i(\hat{x}_i)$) arises from the Hessian of the local cost functions $f_i$, i.e., $B_i(\hat{x}_i)$ in Lemma~\ref{eq:lemma1}. However, in many practical cases, \eqref{eq:phi_matrix} admits an equivalent or sufficient LMI reformulation. We distinguish the following three cases.

\textbf{Case 1:} The function $f_i$ is quadratic. Then, the corresponding Hessian matrix $B_i(\hat{x}_i)$ is constant. Hence, $A_i(\hat{x}_i)$ becomes state independent, i.e., $A_i(\hat{x}_i)=A_i$. We collect all indices $i$, such that $f_i$ is quadratic in the set $\mathcal{N}_1\subset \mathcal{N}$.

\textbf{Case 2:} 
The function $f_i$ is not quadratic, but its Hessian matrix $B_i(\hat{x}_i)$ admits an exact polytopic representation. Then $A_i(\hat{x}_i)$ can be written as a convex combination of state-independent matrices, i.e.,
$A_i(\hat{x}_i) = \sum_{j_i=1}^{q_i} \delta_{j_i,i}(\hat{x}_i)\Delta_{j_i,i}$,
$\delta_{j,i}(\hat{x}_i) \geq 0$,
$\sum_{j_i=1}^{q_i}\delta_{j_i,i}(\hat{x}_i) = 1$,
where $\delta_{j_i,i}(\hat{x}_i)$ are state-dependent convex combination coefficients, $\Delta_{j_i,i}$ are state-independent vertex matrices and $q_i>1$ is the number of vertices.  
We collect the corresponding indices $i$ in $\mathcal{N}_2\subset \mathcal{N}$.

\textbf{Case 3:} The function $f_i$ is not quadratic and no exact polytopic representation of its Hessian matrix $B_i(\hat{x}_i)$ is available. Then we decompose $A_i(\hat{x}_i)$ into a constant and a state-dependent part, i.e., $A_i(\hat{x}_i) = A_{i}^0 + A_i^x(\hat{x}_i)$, with
\begin{align*} 
    A^0_{i} =
    \begin{bmatrix}
        0 & -A_{ii}^\top\\
        A_{ii} & 0
    \end{bmatrix}, 
    \quad 
   A_i^x(\hat{x}_i) =
    \begin{bmatrix}
        - B_i(x_i) & 0\\
        0 & 0
    \end{bmatrix}.
\end{align*}
The indices are collected in
$\mathcal{N}_3=\mathcal{N}\setminus(\mathcal{N}_1\cup\mathcal{N}_2)$.

Let the global vertex index set be
$\mathcal J := \prod_{i\in\mathcal N_2}\{1,\dots,q_i\}$, where each $j \in \mathcal J$ is the tuple $j = (j_i)_{i\in\mathcal N_2}, \ \forall \, j_i \in \{1,\dots,q_i\}$. Define  
$\varsigma := \sum_{i\in\mathcal{N}_3} m_i$, 
$\sigma := \sum_{i\in\mathcal{N}_3} n_i$. Next, we present state independent stability conditions yielding a unified LMI-based formulation that accommodates all three cases. The proof is given in Appendix~\ref{app:proof_corollary}.

\begin{corollary}\label{corollary:LMI_formulation}
Consider the system~\eqref{eq:compact_shifted} with Assumptions~\ref{eq:cost_convexity}-\ref{eq:rank_cond}. 
Let $h_k\ge0$, $d_k\in[0,1)$, $\forall \, k\in\{1,\dots,\rho\}$ and the constant scalar $\varepsilon>0$ be given.
Assume there exist $r \times r$ matrices $Y_{11},Y_{12},Y_{22}$, 
$R_k>0$, $Q_k>0$, $S_k>0$, $S_{12,k}$ for all $k\in\{1,\dots,\rho\}$, 
and block-diagonal matrices $X=\blkdiag(X_i)_{i\in\mathcal N}$ with $X_i\in\mathbb{R}^{r_i\times r_i}$,
$P_2=\blkdiag(P_{2,i})_{i\in\mathcal N}$, where $P_{2,i}\in\mathbb{R}^{r_i\times r_i}$ with $P_{2,i}+P_{2,i}^\top>0$ $\forall \, i \in \mathcal N$ and 
\[
P_{2,i}=
\begin{bmatrix}
w_{11,i}I_{n_i} & W_{12,i}\\
W_{21,i} & W_{22,i}
\end{bmatrix}
\quad 
\forall \, i \in \mathcal N_3,
\]
as well as
$\Omega_1 = \blkdiag(\Omega_{1i})_{i\in\mathcal{N}_3} > 0$, 
$\Omega_2 = \blkdiag(\Omega_{2i})_{i\in\mathcal{N}_3} > 0$,
$\Omega_3 = \blkdiag(\Omega_{3i})_{i\in\mathcal{N}_3} > 0$,
with \(\Omega_{1i}, \Omega_{3i} \in \mathbb{R}^{m_i\times m_i}\) and \(\Omega_{2i} \in \mathbb{R}^{n_i\times n_i}\),
such that~\eqref{eq:jensen_const} and the following matrix inequalities are feasible  
\begin{align}\label{eq:Xi_vertex}
\Xi^{(j)} =
\begin{bmatrix}
\tilde{\Phi}^{(j)} & \Theta \\
\ast & \mathrm{diag}(\Omega_1,\Omega_2,\Omega_3)
\end{bmatrix} < 0
\qquad
\forall \, j \in \mathcal{J},
\end{align}
where 
\begin{align*}
      \tilde{\Phi}^{(j)} =&
\begin{bmatrix}
\phi_{11} & \phi_{12} & \phi_{13} & 0_{r\times \rho r} & 0_{r\times \rho r} \\
\ast & \tilde{\phi}_{22}^{(j)} & \tilde{\phi}_{23}^{(j)} & \phi_{24} & \phi_{25} \\
\ast & \ast & \tilde{\phi}_{33} & 0_{r\times \rho r} & \phi_{35} \\
\ast & \ast & \ast & \phi_{44} & \phi_{45} \\
\ast & \ast & \ast & \ast & \phi_{55}
\end{bmatrix},\\ 
\tilde{\phi}_{22}^{(j)} =& {Y}_{12} + {Y}_{12}^\top + \Psi_1 + {P}_{2}^\top\tilde{\mathcal{A}}^{(j)}+ X  \\
&+ (\tilde{\mathcal{A}}^{(j)})^\top{P}_{2} +X^\top + \sum_{k=1}^{\rho}\left( {S}_{k} + {Q}_{k} - {R}_{k} \right),       \\
    \tilde{\phi}_{23}^{(j)} =& {Y}_{11}-{P}_{2}^\top + \varepsilon(\tilde{\mathcal{A}}^{(j)})^\top{P}_{2}+ \varepsilon X^\top, \\
    \tilde{\phi}_{33} =& -\varepsilon{P}_{2}-\varepsilon{P}_{2}^\top +  \Psi_2 +\sum_{k=1}^{\rho} h_{k}^2 {R}_{k}, \\
    \tilde{\mathcal A}^{(j)}
=&\blkdiag(\tilde{A}_i)_{i\in\mathcal{N}}, \
\tilde{A}_i =
\begin{cases}
A_i, & i \in \mathcal N_1,\\
\Delta_{j_i,i}, & i \in \mathcal N_2,\\
A^0_{i}, & i \in \mathcal N_3,
\end{cases} \notag
\end{align*}
\begin{align}  
 \Psi_1=&\blkdiag(\Psi_{1i})_{i\in\mathcal{N}}, 
\
\Psi_2=\blkdiag(\Psi_{2i})_{i\in\mathcal{N}}, \notag\\
\Psi_{1i}=&
\begin{cases}
\begin{bmatrix}
-2\mu_i w_{11,i}I_{n_i} & 0\\
0 & \Omega_{1i}
\end{bmatrix}, & i\in\mathcal{N}_3,\\
0_{r_i\times r_i}, & i\notin\mathcal{N}_3,
\end{cases} \notag \\
\Psi_{2i}=&
\begin{cases}
\begin{bmatrix}
\Omega_{2i} & 0\\
0 & \Omega_{3i}
\end{bmatrix}, & i\in\mathcal{N}_3,\\
0_{r_i\times r_i}, & i\notin\mathcal{N}_3,
\end{cases} \notag\\
    \Theta =&
\begin{bmatrix}
0_{r \times \varsigma} & 0_{r \times \sigma} & 0_{r \times \varsigma} \\
\varphi_{26} & \varphi_{27} & \varphi_{28} \\
0_{3r \times \varsigma} & 0_{3r \times \sigma} & 0_{3r \times \varsigma} 
\end{bmatrix},  \notag\\
    {\varphi}_{26} =& \blkdiag(\begin{bmatrix}
        \ell_iW_{12,i}^\top & 0_{m_i\times m_i}\end{bmatrix}^\top), \ \forall \, i\in\mathcal{N}_3,\notag\\
    {\varphi}_{27} =& \blkdiag(\begin{bmatrix}
        \varepsilon\ell_iw_{11,i} I_{n_i} & 0_{n_i\times m_i}\end{bmatrix}^\top), \  \forall \,i\in\mathcal{N}_3,\notag\\
    {\varphi}_{28} =& \blkdiag(\begin{bmatrix}
         \varepsilon\ell_iW_{12,i}^\top & 0_{m_i\times m_i}\end{bmatrix}^\top), \ \forall \, i \in\mathcal{N}_3. \notag
\end{align} 
Choose $\mathcal{K}=(P_2^\top)^{-1} X$. Then the origin is uniformly asymptotically stable for all delays $\tau_{k}(t)\in [0, h_{k}]$ with $\left|\dot{\tau}_{k}(t)\right|\leq d_{k}$, $\forall \, k \in \{1,\dots,\rho\}$. Moreover, if the above conditions are feasible for $d_k=1$, respectively $Q_{k}=0_{r\times r}$,  then the origin is uniformly asymptotically stable for all fast varying delays $\tau_{k}(t)\in [0, h_{k}]$, $\forall \, k \in \{1,\dots,\rho\}$. 

\end{corollary}

\begin{remark}
Corollary~\ref{corollary:LMI_formulation} provides a unified, state independent stability condition for the set 
$\mathcal N=\mathcal N_1\cup\mathcal N_2\cup\mathcal N_3$, thereby accommodating heterogeneous cost functions and time-varying delays. 
In contrast to Proposition~\ref{prop:stability_cond}, the resulting conditions are state independent but rely on approximations, and thus tend to be more conservative. 
In particular, no approximation is involved for $\mathcal N_1$, the polytopic representation in $\mathcal N_2$ introduces moderate conservatism, while the Young's inequality-based treatment in $\mathcal N_3$ is generally the most conservative. 
If some subsets $\mathcal N_i$ are empty, the corresponding constraints are removed and the optimization problem reduces accordingly.  
\end{remark}

\begin{remark}\label{remark:minimizing_gains}
    Corollary~\ref{corollary:LMI_formulation} provides sufficient conditions on $\mathcal{K}$ ensuring delay-robust stability without imposing any restrictions on the entries of $\mathcal{K}$. By introducing positive parameters $\alpha_1$, $\alpha_2$ and decision variables $\kappa_X$, $\kappa_P$, one can solve the following optimization problem to minimize the the entries of $\mathcal{K}$~\cite{boyd1994linear}
    \begin{align}
        &\min_{\kappa_{X}, \kappa_{P}} \alpha_1 \kappa_X + \alpha_2 \kappa_P \notag\\ 
        &s.t. \ \eqref{eq:Xi_vertex}, \eqref{eq:jensen_const}, \notag \\ 
        &\begin{bmatrix}
        0.5(P_2+P_2^\top) 
        & I_{r} \\
        \ast & \kappa_P I_{r}
    \end{bmatrix} > 0, \  \begin{bmatrix}
        -\kappa_X I_{r} & X^\top \\
        \ast & -I_{r}
    \end{bmatrix} <0. \label{eq:min_gain}
    \end{align}
    
\end{remark}
\begin{remark}\label{remark:max_delay}
    If  $h_k=\bar{h}$ for $k\in\{1, \ldots, \rho\}$ where $\bar{h}$ is a unified delay upper bound, by solving the following optimization problem, 
    \[ \begin{aligned}
        &\min_{\bar{h}}{-\bar{h}} \quad s.t. \ \eqref{eq:Xi_vertex}, \eqref{eq:jensen_const},
    \end{aligned}
    \]
    we find the maximum allowable delay upper bound (MADUB) and the corresponding gain matrix.
    Although the constraints~\eqref{eq:Xi_vertex} are bilinear matrix 
    inequalities (BMIs), they reduce to LMIs for fixed $\bar{h}$, hence, the problem is quasiconvex in $\bar{h}$ and can be solved numerically via bisection~\cite{boyd2004convex}. 
\end{remark}

\section{Numerical Example}
Consider the optimization problem~\eqref{eq:opt_prob} with the following local cost functions borrowed from~\cite{yang2016distributed}:

\[
\begin{aligned}
f_1(x_1)&=0.9(x_1+1)^2, \ f_2(x_2)=(x_2-4)^2, \\
f_3(x_3)&=0.5x_3^2-1, \ f_4(x_4)=0.6x_4^2+x_4, \\
f_5(x_5)&=(x_5+2)^2,\  f_6(x_6)=0.8x_6^2+2,\\
f_7(x_{7})&=(x_{7}-10)^2,\ 
f_8(x_{8})\!=\! \ln(e^{-0.1x_{8}}\!+\!e^{0.3x_{8}})\!+\!0.9x_{8}^2,\\
f_9(x_{9})&=\sin{\tfrac{x_{9}}{2}}+\tfrac{x_{9}^2}{2}, \ f_{10}(x_{10})=\tfrac{x_{10}^2}{\sqrt{x_{10}^2+9}}+0.6x_{10}^2.
\end{aligned}
\]

The matrix $A$ and the vector $b$ defining the coupling constraints 
are given by
\[
A=\begin{bmatrix}
    I_9 & 0_{9\times 1}
\end{bmatrix}-\begin{bmatrix}
   0_{9\times 1} & I_9 
\end{bmatrix},
\qquad
b = 0_{9\times 1}.
\]

All local cost functions satisfy Assumption~\ref{eq:cost_convexity}. 
Indeed, for $i\in \mathcal{N}_1= \{1,\ldots,7\}$ the functions are quadratic, hence their Hessians are constant. For $i\in \mathcal{N}_2:= \{8,9,10\}$, the Hessians of $f_i$ are state-dependent. However, they admit a polytopic representation. 
Since the matrix $A_i(\hat{x}_i)$ for each $i \in \mathcal{N}_2$, is expressed as a convex combination of two vertices, i.e., $q_i=2$, the matrix $\tilde{\mathcal{A}}^{(j)}$ is described by eight vertices.

To solve the problem, we implement both the standard PDGD~\eqref{eq:PD_delayed} 
and the proposed dynamics~\eqref{eq:proposed_dyn} in MATLAB. 
In both implementations, Agents~1–9 compute both primal and dual variables, 
whereas Agent~10 computes only its primal variable.
Homogenous time-varying communication delays $\tau_k(t)=\tau(t)$ with $\left|\dot{\tau}(t)\right|\leq d=0.1$ are assumed across all channels. 
All optimization problems are solved in MATLAB using Yalmip~\cite{lofberg2004yalmip} with the Mosek~\cite{mosek2019mosek} solver. 

The maximum allowable delay upper bound (MADUB) of the standard PDGD~\eqref{eq:PD_delayed} is computed by solving the conditions of Theorem~3.5 in~\cite{fridman2014introduction} for all vertices via bisection. 
The result shows that the standard PDGD~\eqref{eq:PD_delayed} remains stable for delays up to $\bar{h}=0.372\,\mathrm{s}$ according to Theorem~3.5 in~\cite{fridman2014introduction}.

For the proposed dynamics~\eqref{eq:proposed_dyn}, the MADUB is computed for different choices of the tuning parameter $\varepsilon$ by solving the optimization problem described in Remark~\ref{remark:max_delay} with an additional constraint~\eqref{eq:min_gain}. 
Since LKF-based stability analysis provides only \emph{sufficient conditions} for stability, the resulting delay bounds may be conservative~\cite{fridman2014introduction}. 
In Corollary~\ref{corollary:LMI_formulation}, additional conservatism may arise from restricting \(P_2\) and \(X\) to be block-diagonal matrices and from fixing \(\varepsilon\).
To assess the conservatism of the proposed approach, for each resulting gain matrix \(\mathcal{K}\), we compute the MADUB by using Theorem~3.5 in~\cite{fridman2014introduction}. As the latter conditions are used only for stability analysis, rather than gain synthesis, we do not impose a block-diagonal structure on the decision variables, which may yield less conservative results.

Table~\ref{tab:conservatism_analysis} summarizes the performance analysis. 
The results indicate that for different tuning parameters $\varepsilon$, we obtain a different MADUB by using Remark~\ref{remark:max_delay}. 
We compute the largest MADUB $\bar{h}=1.017\,\mathrm{s}$ with $\varepsilon=1.5$, which is approximately $2.7$ times bigger than the MADUB of the standard PDGD~\eqref{eq:PD_delayed}. 
It can be seen from Table~\ref{tab:conservatism_analysis} that with the gains obtained via Remark~\ref{remark:max_delay}, the proposed dynamics~\eqref{eq:proposed_dyn} is even robust to much larger delays according to Theorem~3.5 in~\cite{fridman2014introduction}. 
Thus, we conclude that overall~\eqref{eq:proposed_dyn} provides significantly improved delay robustness compared to~\eqref{eq:PD_delayed}.

\begin{table}
\caption{Performance Analysis of the MADUBs obtained via the conditions in Remark~\ref{remark:max_delay} and Theorem~3.5 in~\cite{fridman2014introduction}.}
\label{tab:conservatism_analysis}
\centering
\begin{tabular}{c|cc}
\hline
\multirow{2}{*}{\shortstack{Tuning\\parameter $\varepsilon$}} 
& \multicolumn{2}{c}{Maximum allowable delay upper bound ($\bar{h}$)}  \\
& by Remark~\ref{remark:max_delay}& by Theorem~3.5 in~\cite{fridman2014introduction}   \\
\hline
$0.5$ & $0.712\,\mathrm{s}$  & $6.04\,\mathrm{s}$  \\
$1$   & $0.915\,\mathrm{s}$ & $4.72\,\mathrm{s}$ \\
$1.5$ & $1.017\,\mathrm{s}$    & $5.94\,\mathrm{s}$  \\
\hline
\end{tabular}
\end{table}

\section{Conclusion}

In this paper, we proposed a continuous-time augmented PDGD for distributed optimization with equality constraints under heterogeneous, time-varying communication delays. The approach augments the standard PDGD with auxiliary dynamic states coupled to the original dynamics through a gain matrix. An LKF-based analysis yields delay-dependent stability conditions that are reformulated as tractable LMIs for structured gain synthesis. A numerical example demonstrates that the proposed approach significantly improves delay robustness compared to the standard PDGD.

Future work will extend the framework to inequality constraints, rank-deficient equality constraints, and convex cost functions without strong convexity assumptions.



\bibliography{refs_Arxiv}

@article{schiffer2017robustness,
  title={Robustness of distributed averaging control in power systems: Time delays \& dynamic communication topology},
  author={Schiffer, Johannes and D{\"o}rfler, Florian and Fridman, Emilia},
  journal={Automatica},
  volume={80},
  pages={261--271},
  year={2017},
  publisher={Elsevier}
}

@article{schiffer2016stability,
  title={Stability of a class of delayed port-{H}amiltonian systems with application to microgrids with distributed rotational and electronic generation},
  author={Schiffer, Johannes and Fridman, Emilia and Ortega, Romeo and Raisch, J{\"o}rg},
  journal={Automatica},
  volume={74},
  pages={71--79},
  year={2016},
  publisher={Elsevier}
}

@article{qu2018exponential,
  title={On the exponential stability of primal-dual gradient dynamics},
  author={Qu, Guannan and Li, Na},
  journal={IEEE Control Syst. Lett.},
  volume={3},
  number={1},
  pages={43--48},
  year={2018},
  publisher={IEEE}
}

@article{fridman2014introduction,
  title={Introduction to time-delay systems},
  author={Fridman, Emilia},
  journal={Analysis and Control. Birkh{\"a}user},
  volume={75},
  year={2014},
  publisher={Springer}
}

@article{yang2016distributed,
  title={Distributed optimization based on a multiagent system in the presence of communication delays},
  author={Yang, Shaofu and Liu, Qingshan and Wang, Jun},
  journal={IEEE Trans. Syst., Man, Cybern., Syst.},
  volume={47},
  number={5},
  pages={717--728},
  year={2016},
  publisher={IEEE}
}

@article{wang2018distributed,
  title={Distributed optimization for multi-agent systems with constraints set and communication time-delay over a directed graph},
  author={Wang, Dong and Wang, Zhu and Chen, Mingfei and Wang, Wei},
  journal={Inf. Sci.},
  volume={438},
  pages={1--14},
  year={2018},
  publisher={Elsevier}
}

@article{wang2019distributed,
  title={Distributed optimization for resource allocation problems under large delays},
  author={Wang, Xue-Fang and Hong, Yiguang and Sun, Xi-Ming and Liu, Kun-Zhi},
  journal={IEEE Trans. Ind. Electron.},
  volume={66},
  number={12},
  pages={9448--9457},
  year={2019},
  publisher={IEEE}
}

@article{he2004parameter,
  title={Parameter-dependent {Lyapunov} functional for stability of time-delay systems with polytopic-type uncertainties},
  author={He, Yong and Wu, Min and She, Jin-Hua and Liu, Guo-Ping},
  journal={IEEE Trans. Autom. Control},
  volume={49},
  number={5},
  pages={828--832},
  year={2004},
  publisher={IEEE}
}

@article{yang2019survey,
  title={A survey of distributed optimization},
  author={Yang, Tao and Yi, Xinlei and Wu, Junfeng and Yuan, Ye and Wu, Di and Meng, Ziyang and Hong, Yiguang and Wang, Hong and Lin, Zongli and Johansson, Karl H},
  journal={Annual Reviews in Control},
  volume={47},
  pages={278--305},
  year={2019},
  publisher={Elsevier}
}

@article{nedic2018distributed,
  title={Distributed optimization for control},
  author={Nedi{\'c}, Angelia and Liu, Ji},
  journal={Annu. Rev. Control Robot. Auton. Syst.},
  volume={1},
  number={1},
  pages={77--103},
  year={2018},
  publisher={Annual Reviews}
}

@article{liu2019distributed,
  title={A distributed optimization algorithm based on multiagent network for economic dispatch with region partitioning},
  author={Liu, Qingshan and Le, Xinyi and Li, Kaixuan},
  journal={IEEE Trans. Cybern.},
  volume={51},
  number={5},
  pages={2466--2475},
  year={2019},
  publisher={IEEE}
}

@article{dall2013distributed,
  title={Distributed optimal power flow for smart microgrids},
  author={Dall'Anese, Emiliano and Zhu, Hao and Giannakis, Georgios B},
  journal={IEEE Transactions on Smart Grid},
  volume={4},
  number={3},
  pages={1464--1475},
  year={2013},
  publisher={IEEE}
}

@article{xu2017distributed,
  title={A distributed algorithm for resource allocation over dynamic digraphs},
  author={Xu, Yun and Han, Tingrui and Cai, Kai and Lin, Zhiyun and Yan, Gangfeng and Fu, Minyue},
  journal={IEEE Transactions on Signal Processing},
  volume={65},
  number={10},
  pages={2600--2612},
  year={2017},
  publisher={IEEE}
}

@article{kose1956solutions,
  title={Solutions of saddle value problems by differential equations},
  author={Kose, Tairoku},
  journal={Econometrica},
  pages={59--70},
  year={1956},
  publisher={JSTOR}
}

@book{arrow1958studies,
  title={Studies in linear and non-linear programming},
  author={Arrow, Kenneth Joseph and Hurwicz, Leonid and Uzawa, Hirofumi and Chenery, Hollis Burnley and Johnson, Selmer and Karlin, Samuel},
  volume={2},
  year={1958},
  publisher={Stanford University Press Stanford}
}

@article{cherukuri2016asymptotic,
  title={Asymptotic convergence of constrained primal-dual dynamics},
  author={Cherukuri, Ashish and Mallada, Enrique and Cort{\'e}s, Jorge},
  journal={Syst. Control Lett.},
  volume={87},
  pages={10--15},
  year={2016},
  publisher={Elsevier}
}

@article{feijer2010stability,
  title={Stability of primal--dual gradient dynamics and applications to network optimization},
  author={Feijer, Diego and Paganini, Fernando},
  journal={Automatica},
  volume={46},
  number={12},
  pages={1974--1981},
  year={2010},
  publisher={Elsevier}
}

@article{guo2022exponential,
  title={Exponential convergence of primal-dual dynamics under general conditions and its application to distributed optimization},
  author={Guo, Luyao and Shi, Xinli and Cao, Jinde and Wang, Zihao},
  journal={IEEE Trans. Neural Netw. Learn. Syst.},
  volume={35},
  number={4},
  pages={5551--5565},
  year={2022},
  publisher={IEEE}
}

@article{li2020smooth,
  title={Smooth dynamics for distributed constrained optimization with heterogeneous delays},
  author={Li, Mengmou and Yamashita, Shunya and Hatanaka, Takeshi and Chesi, Graziano},
  journal={IEEE Control Syst. Lett.},
  volume={4},
  number={3},
  pages={626--631},
  year={2020},
  publisher={IEEE}
}

@book{boyd2004convex,
  title={Convex optimization},
  author={Boyd, Stephen and Vandenberghe, Lieven},
  year={2004},
  publisher={Cambridge university press}
}

@inproceedings{wang2011control,
  title={A control perspective for centralized and distributed convex optimization},
  author={Wang, Jing and Elia, Nicola},
  booktitle={2011 50th IEEE Conf. Decis. Control Eur. Control Conf. (CDC–ECC)},
  pages={3800--3805},
  year={2011},
  organization={IEEE}
}

@book{boyd1994linear,
  title={Linear matrix inequalities in system and control theory},
  author={Boyd, Stephen and El Ghaoui, Laurent and Feron, Eric and Balakrishnan, Venkataramanan},
  year={1994},
  publisher={SIAM}
}

@article{chang2014distributed,
  title={Distributed constrained optimization by consensus-based primal-dual perturbation method},
  author={Chang, Tsung-Hui and Nedi{\'c}, Angelia and Scaglione, Anna},
  journal={IEEE Trans. Autom. Control},
  volume={59},
  number={6},
  pages={1524--1538},
  year={2014},
  publisher={IEEE}
}

@inproceedings{lofberg2004yalmip,
  title={{YALMIP}: A toolbox for modeling and optimization in {MATLAB}},
  author={Lofberg, Johan},
  booktitle={Proc. IEEE Int. Conf. Robot. Autom. (ICRA), 2004.},
  pages={284--289},
  year={2004},
  organization={IEEE}
}

@article{mosek2019mosek,
  title={MOSEK optimization toolbox for {MATLAB}},
  author={Mosek, ApS},
  journal={User’s guide and reference manual, version},
  volume={4},
  number={1},
  pages={116},
  year={2019}
}

@article{jakovetic2020primal,
  title={Primal-dual methods for large-scale and distributed convex optimization and data analytics},
  author={Jakoveti{\'c}, Du{\v{s}}an and Bajovi{\'c}, Dragana and Xavier, Jo{\~a}o and Moura, Jos{\'e} MF},
  journal={Proc. of the IEEE},
  volume={108},
  number={11},
  pages={1923--1938},
  year={2020},
  publisher={IEEE}
}

@Book{FB-LNS,
  author =    {F. Bullo},
  title =     {Lectures on Network Systems},
  year =      2024,
  edition =   {{1.7}},
  publisher = {Kindle Direct Publishing},
  ISBN =      {978-1986425643},
  url =       {https://fbullo.github.io/lns},
}
 \bibliographystyle{ieeetr}
\appendix
\subsection{Proof of Proposition~\ref{prop:stability_cond}} \label{app:proof_1}
Inspired by Proposition 5.3 in~\cite{fridman2014introduction}, consider the LKF 
\begin{subequations}\label{eq:LKF}
\begin{align}
    V= & V_1 +\sum_{k=1}^\rho \left(V_{2,k} + V_{3,k} + V_{4,k}\right), \\
    V_{1} =& \left[\begin{array}{c}
         \tilde{z} \\
         \tilde{u} 
    \end{array} \right]^\top \left[\begin{array}{cc}
         {Y}_{11} & {Y}_{12} \\
         \ast & {Y}_{22}
    \end{array} \right] \left[\begin{array}{c}
         \tilde{z} \\
         \tilde{u} 
    \end{array} \right], \\
    V_{2,k} =& \int_{t-h_{k}}^t \tilde{z}^\top(s) {S}_{k}\tilde{z}(s)ds, \\
    V_{3,k} =&\int_{t-\tau_{k}}^t \tilde{z}^\top(s) {Q}_{k}\tilde{z}(s) ds, \\
    V_{4,k} =&h_k \int^{0}_{-h_{k}}\int^t_{t+\theta}\dot{\tilde{z}}^\top(s) {R}_{k}\dot{\tilde{z}}(s) dsd\theta.
\end{align}
\end{subequations}
The derivative of the LKF~\eqref{eq:LKF} along the solutions of the system~\eqref{eq:compact_shifted} yields
\begin{subequations}\label{eq:derivative_V}
\allowdisplaybreaks
\begin{align}
   \dot{V} =& \dot{V}_{1}+\sum_{k=1}^\rho \left(\dot{V}_{2,k} + \dot{V}_{3,k} + \dot{V}_{4,k}\right), \label{eq:V_dot} \\
        \dot{V}_{1} =& 2\tilde{z}^\top {Y}_{11}\dot{\tilde{z}} +2\tilde{u}^\top {Y}_{12}^\top\dot{\tilde{z}} -2 \tilde{z}^\top {Y}_{12}\tilde{u} \\&+2\tilde{z}^\top Y_{12}\tilde{z} 
    -2\tilde{u}^\top {Y}_{22}\tilde{u} +2\tilde{u}^\top {Y}_{22}\tilde{z}, \notag \\
    \dot{V}_{2,k} =& \tilde{z}^\top {S}_{k} \tilde{z} - \tilde{z}_{h_{k}}^\top{S}_{k}\tilde{z}_{h_{k}}, \\ 
    \dot{V}_{3,k} = &  \tilde{z}^\top {Q}_{k} \tilde{z} - (1-\dot{\tau}_{k})\tilde{z}_{\tau_{k}}^\top{Q}_{k}\tilde{z}_{\tau_{k}} \notag\\
    \leq & \tilde{z}^\top {Q}_{k} \tilde{z} -(1-d_{k}) \tilde{z}_{\tau_{k}}^\top{Q}_{k}\tilde{z}_{\tau_{k}},  \\ 
    \dot{V}_{4,k} = & h_{k}^2\dot{\tilde{z}}^\top {R}_{k} \dot{\tilde{z}} - h_{k}\int_{t-h_{k}}^{t}\dot{\tilde{z}}^\top(s){R}_{k}\dot{\tilde{z}}(s)ds \label{eq:dot_V4}.
\end{align}
\end{subequations}
For the second term in~\eqref{eq:dot_V4}, we apply Jensen's inequality together with Lemma 3.4 in \cite{fridman2014introduction}, which yields
\begin{align}
   & - h_{k}\int_{t-h_{k}}^{t}\dot{\tilde{z}}^\top(s){R}_{k}\dot{\tilde{z}}(s)ds \nonumber \\
   &= - h_{k}\int_{t-h_{k}}^{t-\tau_{k}}\dot{\tilde{z}}^\top(s){R}_{k}\dot{\tilde{z}}(s)ds - h_{k}\int_{t-\tau_{k}}^{t}\dot{z}^\top(s){R}_{k}\dot{\tilde{z}}(s)ds \nonumber \\
   &\leq -\begin{bmatrix}
       \tilde{z} - \tilde{z}_{\tau_{k}} \\
        \tilde{z}_{\tau_{k}} - \tilde{z}_{h_{k}}
   \end{bmatrix}^\top \left[\begin{array}{cc}
      {R}_{k}  & {S}_{12,{k}}  \\
       \ast & {R}_{k}
   \end{array} \right] \begin{bmatrix}
       \tilde{z} - \tilde{z}_{\tau_{k}} \\
        \tilde{z}_{\tau_{k}} - \tilde{z}_{h_{k}}
   \end{bmatrix}. \notag
\end{align}
Next, we apply the descriptor method (see, e.g., \cite{fridman2014introduction}), i.e., we add the following term to the right hand side of~\eqref{eq:V_dot} 
\begin{equation}\label{eq:descriptor_method}
    \begin{aligned}
        0\! = \! 2\big(\tilde{z}^\top P_2^\top \! \!+\!\dot{\tilde{z}}^\top P_3^\top  \big)\! \big(\!
-\!\dot{\tilde{z}} \!+\! (\mathcal{A}(\hat{x})\! + \mathcal{K})\tilde{z}
\!+\!  \sum_{k=1}^{\rho} \mathcal{T}_k \tilde{z}_{\tau_{k}}
\!-\! \mathcal{K}\tilde{u}
\big),
    \end{aligned}
\end{equation}
where $P_2, P_3\in \mathbb{R}^{r\times r}$ are auxiliary matrices.
The descriptor method introduces terms with $P_2^\top \mathcal{K}$ and $P_3^\top \mathcal{K}$. To synthesize $\mathcal{K}$ with a block structure, we set $P_3 = \varepsilon P_2$, where $\varepsilon$ is a tuning scalar and define the auxiliary variable $X = P_2^\top \mathcal{K}$.
Then, we get 
$\dot{V} \leq \xi^\top {\Phi}(x) \xi,$
where 
$\xi = \col(
\tilde u,\tilde z,\dot{\tilde z},
\tilde z_{h_{1}},\ldots,\tilde z_{h_{\rho}},
\tilde z_{\tau_{1}},\ldots,\tilde z_{\tau_{\rho}}),$
and $\Phi(x)$ is defined in~\eqref{eq:phi_matrix}. Since by assumption \eqref{eq:jensen_const} and $\Phi(x)<0$ are satisfied, we have $\dot{V} \leq -\alpha \|\col(\tilde{z}, \tilde{u})\|^2$ for some $\alpha>0$,
which proves uniform asymptotic stability of the origin of the system~\eqref{eq:compact_shifted}~\cite{fridman2014introduction}. Moreover, since $P_{2}+P_{2}^\top>0$, the gain matrix can be recovered as $\mathcal{K}=(P_2^\top)^{-1}X$.
\subsection{Proof of Corollary~\ref{corollary:LMI_formulation}}\label{app:proof_corollary}
The state-independent local matrices $A_i$ for $i\in\mathcal{N}_1$, $A^0_i$ for $i\in\mathcal{N}_3$ and the state-dependent local matrices $A_i(\hat{x}_i)$ for $i\in\mathcal{N}_2$ can be written as a global block diagonal matrix by using the polytopic structure such that $\tilde{\mathcal{A}}(\hat{x})=\sum_{j\in \mathcal{J}}\left(\prod_{i\in\mathcal N_2} \delta_{j_i,i}(\hat{x}_i)\right)\tilde{\mathcal{A}}^{(j)}$. 
Now, we focus on the term $A^x_i(\hat{x}_i)$ in Case 3. The state-dependent terms arise through~\eqref{eq:descriptor_method} in Appendix~\ref{app:proof_1} with $P_3=\varepsilon P_2$, and must be bounded for all $i \in \mathcal{N}_3$. 
Using the bounds in Lemma~\ref{eq:lemma1} and applying Young's inequality~\cite{boyd2004convex} 
yields
\begin{align}
    -2& \tilde{z}_i^\top P_{2i}^\top A^x_i(\hat{x}_i) \tilde{z}_i -2 \dot{\tilde{z}}_i^\top \varepsilon P_{2i}^\top A^x_i(\hat{x}_i) \tilde{z}_i \notag\\
    \leq& \tilde{z}_i^\top \Psi_{1i} \tilde{z}_i +\dot{\tilde{z}}_i^\top \Psi_{2i} \dot{\tilde{z}}_i + \ell_i^2 \tilde{x}_i^\top W_{12,i} \Omega_{1i}^{-1}  W_{12,i}^\top \tilde{x}_i  \label{eq:youngs_ineq}\\
    & +\varepsilon^{2} \ell_i^2 w_{11,i}^2 \tilde{x}_i^\top  \Omega_{2i}^{-1} \tilde{x}_i +  \varepsilon^2 \ell_i^2 \tilde{x}_i^\top  W_{12,i} \Omega_{3i}^{-1}  W_{12,i}^\top \tilde{x}_i. \notag
\end{align}
Instead of $\mathcal{A}(\hat{x})$ in~\eqref{eq:phi_matrix}, by writing $\tilde{\mathcal{A}}(\hat{x})$, $\Psi_1$, $\Psi_2$ as defined in~\eqref{eq:Xi_vertex} and applying the Schur complement~\cite{boyd1994linear} to the last three terms in~\eqref{eq:youngs_ineq}, the matrix inequality
$\Phi(\hat{x})<0$ in~\eqref{eq:phi_matrix} can be rewritten
as $\Xi(\hat{x})=\sum_{j\in \mathcal{J}}\left(\prod_{i\in\mathcal N_2} \delta_{j_i,i}(\hat{x}_i)\right)\Xi^{(j)}<0$.
Since the coefficients $\delta_{j_i,i}(\hat{x}_i)$ are state dependent,
feasibility with a common set of decision variables
guarantees that the stability conditions need only to be verified
at the vertices ($\Xi^{(j)}$) of the polytope
\cite{schiffer2016stability, he2004parameter}. Uniform asymptotic stability of the origin then follows by repeating the arguments from the proof of Proposition~\ref{prop:stability_cond}.

\end{document}